\documentclass[runningheads]{llncs}
\usepackage{graphicx}
\usepackage[utf8]{inputenc} 
\usepackage[T1]{fontenc}    
\usepackage{hyperref}       
\usepackage{url}            
\usepackage{booktabs}       
\usepackage{amsfonts,amssymb,amsmath,latexsym}     
\usepackage{nicefrac}       
\usepackage{microtype}      
\usepackage{breakcites}

\begin{document}
\title{Causal Games and Causal Nash Equilibrium}
%
%
\author{Mauricio Gonzalez-Soto\inst{1}\orcidID{0003-2668-9013} \and
Luis E. Sucar\inst{1} \orcidID{0002-3685-5567} \and
Hugo Jair Escalante\inst{1,2} \orcidID{0003-4603-3513}}
\authorrunning{M. Gonzalez-Soto et al.}
%
\institute{Coordinación de Ciencias Computacionales,\\  Instituto Nacional de Astrofísica Óptica y Electrónica, Tonanzintla, Mexico. \and
Departamento de Computación,\\ Centro de Investigaciones y Estudios Avanzados del IPN, Zacatenco, Mexico\\
\email{ \{mauricio, esucar, hugojair\}@inaoep.mx}}

\maketitle              
\begin{abstract}
Classical results of Decision Theory, and its extension to a multi-agent setting: Game Theory, operate only at the \textit{associative} level of information; this is, classical decision makers only take into account probabilities of events; we go one step further and consider \textit{causal} information: in this work, we define Causal Decision Problems and extend them to a multi-agent decision problem, which we call a causal game. For such games, we study belief updating in a class of strategic games in which any player's action \textit{causes} some consequence via a causal model, which is unknown by all players; for this reason, the most suitable model is Harsanyi's Bayesian Game. We propose a probability updating for the Bayesian Game in such a way that the knowledge of any player in terms of probabilistic beliefs about the causal model, as well as what is caused by her actions as well as the actions of every other player are taken into account. Based on such probability updating we define a Nash equilibria for Causal Games.
\end{abstract}

\section{Introduction}
Causal reasoning is a constant element in our lives as it is human nature to constantly ask \textit{why}. Looking for causes is an everyday task and, in fact, causal reasoning is to be found at the very core of our minds~\cite{waldmann2013causal,danks2014unifying}. It has been argued that the brain itself is a causal inference machine which uses \textit{effects} to figure out \textit{causes} in order to actively engage with the world~\cite{friston2010free,clark2015surfing,lake2017building}. 

An important aspect of acting in the world is being able to make decisions under uncertain conditions~\cite{danks2014unifying,lake2017building}. In their seminal work~\cite{von1944theory}, von Neumann and Morgenstern answered how to make choices if \textit{rational} preferences are assumed and the decision maker knows the stochastic relation between actions and outcomes: maximize expected utility. If such relation is unknown, then J.L. Savage showed in \cite{savage1954the} that a rational decision maker must choose \textit{as if} is maximizing her expected utility with respect to a \textit{subjective} probability distribution. 

The previous results of von Neumann-Morgenstern and Savage provide formal criteria for decision making if rationality is assumed and information about the environment is considered only at the \textit{asociative} (i.e., probabilistic) level. These criteria are the basis for many of the techniques used in Artificial Intelligence; for example, Reinforcement Learning algorithms learn \textit{optimal policies} that satisfy the Bellman Equations~\cite{sutton1998reinforcement,Puterman:1994:MDP:528623}; therefore, any action prescribed by an optimal policy achieves the maximum expected utility as shown in~\cite{webb2007game}.

Several studies have considered how human beings use causal information when making decisions with uncertain outcomes. It is known that humans tend to prefer causal information over purely probabilistic data \cite{tversky1980causal}; and, in fact, it is shown in \cite{hagmayer2009decision} that acting in the world is conceived by human beings as \textit{intervening} on it; Therefore, it does not come up as a surprise that humans are able to learn and use causal relations while making single choices as well as in sequential decisions as shown in~\cite{danks2014unifying,sloman2006causal,Garcia-Retamero2006,nichols2007decision,lagnado2007beyond,hagmayer2008causal,meder2010observing,hagmayer2013repeated,rottman2014reasoning,hagmayer2017causality}.

Decision problems faced by a rational agent usually involve the decisions made by other agents as well as other, possibly unknown, factors. As seen in several applications, interactive reasoning is a fundamental aspect of human every-day reasoning and it should be addressed by any intelligent agent as argued in \cite{lake2017building}. We consider to be of interest the multi-agent setting for causal decision making; for this reason, we consider the interaction of several rational and causal-aware decision makers whose decisions affect each other. 

Game Theory \cite{osborne1994course} deals with situations in which several \textit{rational} decision makers, or players, \textit{interact} while pursuing some well-defined objective; the case in which decision makers make a choice simultaneously without knowing the choice made by the other players is called a \textit{strategic game}; a well-known strategic game is the famous \textit{prisoners' dilemma} in which two detainees must choose between confessing or remaining silent and both know the consequences of any combination of actions, what is ignored by each player is the decision made by the other. 

When players ignore both the actions made by other players as well as the knowledge that made them choose a certain action, is called a game with incomplete information, or a \textit{Bayesian Game} which was introduced by Harsanyi in \cite{harsanyi1967games1, harsanyi1968games2, harsanyi1968games3}. In this work we will use the Bayesian Game model in order to study what happens when several decision makers have certain knowledge about an environment which is controlled by some, unknown but fixed, causal mechanism. We will first study one-player games, or decision problems, in which the player's actions \textit{cause} some consequence according to some unknown causal model; for this case, we will provide a rational choice criterion which will serve us to define a Nash equilibrium in Causal Games.
\section{Causation and Classical Decision Problems}
\subsection{Causation}
The notion of causation deals with regularities found in a given environment which are stronger than probabilistic (or associative) relations in the sense that a causal relation allows for evaluating a change in the \textit{consequence} given that a change in the \textit{cause} is performed, while probabilistic relations only capture patterns that appear on observed data. For example, when training only on observed samples $(x,y)$, a Bayesian Network can be equally trained as $X \to Y$ or $Y \to X$, see \cite{bengio2019metatransfer} Apendix A for a theoretical argument.

We adopt here the \textit{manipulationist} interpretation of Causality as expressed by Woodward in \cite{woodward2005making}. The main paradigm is clearly expressed in \cite{campbell1979quasi} as \textit{manipulation of a cause will result in a manipulation of the effect}. Consider the following example from \cite{woodward2005making}: manually forcing a barometer to go down won't cause a storm, whereas the occurrence of a storm will cause the barometer to go down. 

We restrict ourselves to probabilistic causation and adopt the formal definition of Causality given in \cite{spirtes2000causation}; i.e., a stochastic relation between events which is \textit{irreflexive, antisymmetric} and \textit{transitive}; such formal definition is encompassed by the manipulationist interpretation. Similar descriptions of the manipulationist approach can be found in \cite{holland1986statistics} and \cite{freedman1997association}. Causal inference tools, such as Pearl's do-calculus, stated in \cite{pearl2009causality}, allows to find the effect of an intervention in terms of probabilistic information when certain conditions are met. For what remains, we assume the \textit{causal axioms} found in \cite{spirtes2000causation} with the condition known as \textit{causal sufficiency}.
\subsection{Classical Decision Theory}
Classical decision making consist of a set $\mathcal{A}$ of available options to a rational decision maker, and a family $\mathcal{E}$ of uncertain events which will affect the consequence of the action made by the decision maker; any knowledge by the decision maker of such uncertainties is available only at the associative, or probabilistic, level of information. We now state the formal framework for classical decision making as we will use it in order to build upon the causal version of it:
\begin{definition}
An uncertain environment is the tuple $(\Omega, \mathcal{A},\mathcal{C},\mathcal{E})$. Where $\mathcal{A}$ is a non-empty set of available actions, $\mathcal{C}$ a set of consequences and $\mathcal{E}$ an algebra of events over $\Omega$. 
\end{definition}
When we consider the preferences of some decision maker over the set of consequences of some uncertain environment we have a Decision Problem under Uncertainty:
\begin{definition}
A Decision Problem under Uncertainty is an uncertain environment $(\Omega, \mathcal{A},\mathcal{C},\mathcal{E})$ plus a preference relation $\succeq$ defined over $\mathcal{C}$ 
\end{definition}
\subsection{Causal Environments and Causal Decision Problems}
In this section we define a Causal Environment to be an \textit{uncertain environment} for which there exists a Causal Graphical Model (CGM) $\mathcal{G}$ which controls the environment. Details on CGMs can be found in \cite{koller2009probabilistic}.
\begin{definition}
A Causal Environment is a tuple $(\Omega, \mathcal{A},\mathcal{G},\mathcal{C},\mathcal{E})$ where $(\Omega, \mathcal{A},\mathcal{C},\mathcal{E})$ is an uncertain environment and $\mathcal{G}$ is a CGM such that the set of variables of $\mathcal{G}$ correspond to the uncertain events in $\mathcal{E}$.
\end{definition}
\subsection{Rational choice in Causal Environments}
Consider a decision maker who knows that any action she takes will \textit{cause} a certain action, but she does not explicitly knows the form of such causal relation, she only have probabilistic beliefs about such relation. We define in this section a formal framework for studying such situations.
\begin{definition}
We define a Causal Decision Problem (CDP) as $(\mathcal{A}, \mathcal{G},\mathcal{E},\mathcal{C},\succeq)$ where $(\mathcal{A}, \mathcal{G},\mathcal{E},\mathcal{C})$ is a Causal Environment and $\succeq$ is a preference relation. 
\end{definition}
For the CGM in a given CDP we will distinguish two particular variables: one corresponding to the available actions, and one corresponding to the caused outcome. We are considering that only one variable can be intervened upon and that the values of such variable represent the actions available to the decision maker; i.e., the value forced upon such variable under an intervention represents the action taken by the decision maker. The intuition behind the definition of a Causal Decision Problem is this: a decision maker chooses an action $a \in \mathcal{A}$, which is automatically inputed into the model $\mathcal{G}$, which outputs the \textit{causal outcome} $c \in \mathcal{C}$. We say a CDP is \textit{finite} if the set $\mathcal{A}$ is finite. We now provide a decision criterion for rationally choosing in a Causal Decision Problem.
\begin{theorem}{\label{causal_savage}}
In a finite Causal Decision Problem  $(\mathcal{A}, \mathcal{G},\mathcal{E},\mathcal{C},\succeq)$, where $\mathcal{G}$ is a Causal Graphical Model, we have that the preferences $\succeq$ of a decision maker are Savage-rational if and only if there exists a probability distribution $P_C$ over a family $\mathcal{F}$ of causal models such that for $a,b \in \mathcal{A}$:
\begin{eqnarray*}
a \succeq & b \textrm{ if and only if }&\\
 \sum_{c \in \mathcal{C}} u(c) \left( \sum_{g \in \mathcal{F}} P_g(c | do(a))P_C(g) \right)\\ &\geq & \\ \sum_{c \in \mathcal{C}}  u(c) \left( \sum_{g \in \mathcal{F}} P_g(c | do(b))P_C(g) \right),
\end{eqnarray*}
where $P_g$ is the probability distribution associated with the causal model $g$. 
\end{theorem}
\begin{proof}
The decision maker is facing an environment in which any action she takes will stochastically cause an outcome $c \in \mathcal{C}$. For this reason, the decision making is facing a very particular case of decision making under uncertainty. Assuming rationality, we invoke Savage's Theorem~\cite{savage1954the,kreps1988choice,gilboa2009decision} to obtain a utility function $u^S$ and a probability measure $P^S$ which satisfy that the preference relation is represented by the expectation of $u^S$ with respect to $P^S$. 

In such a causal environment, the CGM $\mathcal{G}$ contains all of the information which connects actions, uncertain events and outcomes, and noting that we can identify any action $a$ with $\{ c_j | E_j : j \in J  \}$ where $J$ a countable set of indexes~\cite{bernardo2000bayesian} we have that:
\[\mathbb{E}_{P^S}[u(c)] = \sum_{j \in J} u(c_j)P^S(E_j).\]

For each action $a=\{ c_j | E_j : j \in J \}$, $P^S(E_j)$ is the probability of \textit{causing} consequence $c_j$ by choosing action $a$. In order for the decision maker to find the probability of a certain consequence $c_j$ given that an action $a$ is performed then she must have in mind a single causal model $g$ and a way to assign probabilities over a family of causal models; i.e., the uncertainty component $P^S(E_j)$ is formed by two parts: a distribution $P_C$ which represents the degree of belief of the decision maker about a specific model $g$ being the true one, and within $g$, a distribution $P_g$ used to calculate the probability of causing some consequence $c_j$ given that action $a$ is chosen. Using the Caratheodory Extension Theorem~\cite{ash2000probability} a probability measure  $P_C$ whose support is a sufficiently general family of causal models $\mathcal{F}$ can be shown to exist. For $g \in \mathcal{F}$, the decision maker considers $g$ to be the true causal model with probability $P_C(g)$, and within $g$, we use the classical von Neumann-Morgenstern Theorem in order to obtain the best action (see section 4.1 of \cite{pearl2009causality} for details). Let $P_g$ the probability distribution associated with the causal model $g$. 

Then:
\begin{eqnarray}
\mathbb{E}_{P^S}[u(c)] &=& \sum_{j \in J} u(c_j)P^S(E_j)\\
                                      &=& \sum_{j \in J} u(c_j) \left( \sum_{g \in \mathcal{F}} P_g(c_j | do(a))P_C(g) \right).
\end{eqnarray}
We have shown what is the expected utility for some action $a \in \mathcal{A}$, and by Savage's Theorem the result follows. 
\end{proof}
\subsection{Interpretation}
Theorem \ref{causal_savage} says that a decision maker who faces a Causal Decision Problem is considering a probability distribution $P_C$ over a family $\mathcal{F}$ and, within each structure, using the term $P_g(c|do(a))$ in order to find the probability of obtaining a certain consequence given that the intervention $do(a)$ is performed; in this way, the optimal action $a^\ast$ is given by:
\begin{equation}
a^\ast = \textrm{ argmax }_{a \in \mathcal{A}}  \sum_{c \in \mathcal{C}} u(c) \left( \sum_{g \in \mathcal{F}} P_g(c | do(a))P_C(g) \right). 
\end{equation}
We note that $a^\ast$ is obtained by taking into account the utility obtained by every possible consequences weighted using both the probability of causing such action within a specific causal model $g$ and the probability that the decision maker assign to such $g \in \mathcal{F}$.

We are considering a \textit{normative} interpretation for Theorem \ref{causal_savage} according to which a decision maker must use any causal information in order to obtain the best possible action. Such action must be obtained by considering the \textit{beliefs} of the decision maker about the causal relations that hold in her environment (the distribution $P_C$), how such relations could produce the best action when considered \textit{as if} they were true (distribution $P_g$), and the satisfaction (utility $u$) produced by the consequences of actions~\cite{koller2009probabilistic}.
\section{Classical strategic games}
A \textit{strategic game} is a model of a situation in which several players must take an action and afterwards they will be affected both by the outcome of their own action as well as the actions of the other players. In a strategic game it is assumed that no player knows the action taken by any other players; this is, 
\begin{definition}
a \textbf{strategic game} (\cite{osborne1994course}) consists of:
\begin{itemize}
\item A finite set $N$ of $n$ players.
\item For each player, a nonempty set $A_i$ of available actions.
\item For each player, a preference relation $\succeq_i$ defined over $A= A_1 \times \cdots A_n$.
\end{itemize}
\end{definition}
\begin{definition}
A Nash equilibrium of a strategic game $G=(N,(A)_{i \in N},(\succeq_i)_{i \in N})$ is a vector of strategies $a^\ast=(a_1,a_2,...,a_n)$ such that
\[ (a^\ast_{-i},a^\ast_i) \succeq_i (a^\ast_{-i}, b_i) \textrm{ for all } b_i \in A_i, \]
where $a_{-i}=(a_1,...,a_{i-1}, a_{i+1},...,a_n)$.
\end{definition}
This is, in a Nash equilibrium no player can find a better action given the actions taken by the rest of the players. We adopt here the \textit{deductive} interpretation of an equilibrium, according to which an equilibrium results from rationality principles~\cite{binmore1987modeling,binmore1988modeling}.

\section{Causal Games}
In this section, we define a \textit{causal strategic game} as a strategic game within a causal environment; this is, consider a \textit{strategic} game between $N$ rational players who are situated in a causal environment. We assume that it is \textit{common knowledge} the causal nature of the environment as well as the rationality assumption for each player. We also assume that the causal mechanism, which represented by a Causal Graphical Model $\mathcal{G}$, remains fixed and it is unknown for each player. In this game, players ignore the actions taken by any other player, and since the causal model which controls the environment is unknown by every player, then players also ignore the information that players will use in order to take their respective actions: for this reason, we will work within the framework of \textit{bayesian games}.
\subsection{Bayesian Games}
Strategic games are games in which no player knows the action taken by the other players; we now consider a type of game in which no player knows both the actions taken by any other player, nor the \textit{private information} that made each player to take any action. Such model is called a \textit{Bayesian Game}, introduced in \cite{harsanyi1967games1,harsanyi1968games2,harsanyi1968games3}
\begin{definition}
A Bayesian strategic game(\cite{osborne1994course}), consists of:
\begin{itemize}
\item A finite set $N$ of players.
\item A finite set $\Omega$ of \textit{states of nature}.
\item For each player, a nonempty set $A_i$ of actions.
\item For each player, a finite set $T_i$ and a function $\tau_i : \Omega \mapsto T_i$ the signal function of the player
\item For each player, a probability measure $p_i$ over $\Omega$ such that $p_i (\tau^{-1}_i (t_i))>0$ for all $t_i \in T_i$.
\item A preference relation $\succeq_i$ defined over the set of probability measures over $A \times \Omega$ where $A= A_1 \times \cdots A_n$
\end{itemize}
\end{definition}

\subsection{Bayesian Causal Games}
In this section, we consider a \textit{strategic game} between $N$ rational players who are situated in a causal environment. A game is a model of a situation in which several players must take an action and afterwards they will be affected both by the outcome of their own action as well as the actions of the other players. In a strategic game it is assumed that no player knows the action taken by any other players; we also assume that the causal mechanism, which represented by a Causal Graphical Model $\mathcal{G}$, remains fixed and it is unknown for each player. 

In this game, players ignore the actions taken by any other player, and since the causal model which controls the environment is unknown by every player, then players also ignore the information that players will use in order to take their respective actions: strategic games of this type are called \textit{Bayesian Game}, introduced in \cite{harsanyi1967games1,harsanyi1968games2,harsanyi1968games3}. In the games we will consider, the uncertainty of every player consists of two levels: on a first level, the true causal model $\mathcal{G}$; on a second level, what an action $do(a)$ causes if a certain CGM $\omega$ is considered to be the causal model. 

We will consider the set $\Omega$ to be a family of possible causal models; in this way,  $\omega \in \Omega$ being the true state of nature fixes a causal model which controls the environment in which the players make their choices. In classical Bayesian games, once $\omega \in \Omega$ is realized as the true state, then each player receives a signal $t_i=\tau_i (\omega)$ and the posterior belief $p_i(\omega | \tau^{-1}_i (t_i) )$ given by $p_i(\omega) / p_i (\tau^{-1}_i (t_i))$ if $\omega \in \tau^{-1}_i (t_i)$. In the case for causal bayesian games, we must consider both the probability $p_i$ of $\omega$ being the true state as well as the probability $p^\omega_i$ of observing a certain consequence when doing some action $a_i$ if $\omega$ is the true model.

Following \cite{osborne1994course}, we define a new game $G^\ast$ in which its players are all of the possible combinations $(i, t_i) \in N \times T_i$, where the possible actions for $(i,T_i)$ is $A_i$. We see that fixing a player $i \in N$, the posterior probability $p(\omega | \tau^{-1}_i (t_i))$ induces a lottery over the pairs $(a^\ast(j,\tau_j(\omega)))_j,\omega)$ for some other $j \in N$. This lottery assigns to $(a^\ast(j,\tau_j(\omega)))_j,\omega)$ the probability $p_i(\omega) / p_i (\tau^{-1}_i (t_i))$ if $\omega \in \tau^{-1}_i (t_i)$. The classical Bayesian game will simply call a Nash equilibrium for the game $G^\ast$ a Nash equilibrium of the original game; but we have the second level of uncertainty: the consequences caused by some action $a$ through a causal model $\omega$. We notice that the posterior probability itself induces a probability distribution defined over \textit{actions} for each player once a \textit{desired consequence} is fixed, this distribution, according to Theorem \ref{causal_savage} is given by $p^\omega_i (c | do(a^\ast_i), a^\ast_{-i}) p_i(\omega | \tau^{-1}_i (t_i))$. This motivates the following definition of a \textit{Causal Nash equilibrium}.
\subsection{Causal Nash Equilibrium}
For each player $i \in N$ in the strategic game, we define the following probability distribution over consequences:
\begin{equation}{\label{causal_ut}}
p^a_i (c) =  p^\omega_i (c | do(a_i), a_{-i}) p_i(\omega)\textrm{ for } a \in A=A_1 \times \cdots \times A_N.
\end{equation}
where $p^\omega_i$ is the probability of causing a certain consequence within a causal model $\omega$ and $p_i$ are the player's \textit{posterior beliefs} about the causal model that controls the environment, and $do()$ is the well known intervention operator from \cite{pearl2009causality}. We now define:
\begin{equation}
u^C_i (a) = \sum_{c \in C}  u_i(c) p^a_i (c) \textrm{ for } a \in A=A_1 \times \cdots \times A_N.
\end{equation}
Notice that $u^C_i$ evaluates an action profile $a \in A$ in terms of: The knowledge about the causal model of each player represented by $p_i$, which allows each player to evaluate the probability of causing outcomes in terms of actions by using the $do$ operator as well as the other actions taken by the other players, given by $a_{-i}$ and the preferences of each player $u_i$. Using this new function, we define the equilibrium for a strategic game with causal information and Bayesian players as:
\begin{definition}
A Nash equilibrium for this \textit{causal strategic game} is an action profile $a^\ast \in A$ if and only if
\begin{equation}
 u^C_i(a^\ast) \geq u^C_i(a_i, a^\ast_{-i}) \textrm{ for any other } a_i \in A_i. 
 \end{equation}
\end{definition}
This is, an action profile is a Nash equilibrium if and only if each player uses her current knowledge about the causal model of the environment in order to (causally) produce the best possible outcome given the actions taken by the other players. The existence of the Causal Nash Equilibrium is guaranteed if every $A_i$ is a nonempty compact convex set in some $\mathbb{R}^n$ and if the preference relation induced by $u^C_i$ is continuous and quasi-concave.

\section{Conclusion}
We have studied Decision Making under uncertainty in the case where a Causal Graphical Model is responsible for producing an outcome given an action (intervention) of the decision maker. We have provided a rational decision making criterion for the case in which the decision maker does not know the causal model, but has probabilistic beliefs about possible models. 

Using our decision making result, and taking as a basis Harsanyi's model of a Bayesian Game in which every player has incomplete information about both the actions taken by other players as well as the information that made each player take his action we have been able to provide a definition of a Causal Nash Equilibrium in which every player is aware that there exists a Causal Mechanism that will produce some consequence once he takes an action. 

Our decision making result (i.e., Theorem \ref{causal_savage}), besides motivating the Causal Nash Equilibrium, also provides an optimality criterion for learning algorithms in causal settings such as those presented in~\cite{lattimoreNIPS2016,sen2017identifying,gonzalez2018playing}. Our definition of Causal Equilibrium takes into account classical game theory through the incorporation of the classical von Neumann-Morgenstern utility function as well as the fundamental notion in Causation of Pearl's $do$ operator. 

We hope this works contributes to recent efforts of giving Causation its well deserved place in Artificial Intelligence as well as motivating further research in computational aspects of Causal Decision Theory.

\bibliographystyle{ieeetr}
\bibliography{Bibliografia.bib}
\end{document}